\documentclass[journal,onecolumn,12pt]{IEEEtran}
\usepackage{pslatex} 
\usepackage{amsfonts,color,morefloats}
\usepackage{amssymb,amsmath,latexsym,amsthm}

\newcommand{\wt}{{\mathrm{wt}}}

\newcommand{\cS}{{\mathcal{S}}}
\newcommand{\tr}{{\mathrm{Tr}}}
\newcommand{\gf}{{\mathrm{GF}}}

\newcommand{\C}{{\mathcal{C}}}

\newcommand{\cP}{{\mathcal{P}}}

\newcommand{\bc}{{\mathbf{c}}}
\newcommand{\bg}{{\mathbf{g}}}
\newcommand{\bu}{{\mathbf{u}}}

\newcommand{\bt}{{\mathbf t}}
\newcommand{\bx}{{\mathbf x}}
\newcommand{\by}{{\mathbf y}}

\newcommand{\bh}{{\mathbf h}}

\newtheorem{theorem}{Theorem}
\newtheorem{lemma}[theorem]{Lemma}

\newtheorem{corollary}[theorem]{Corollary}

\newtheorem{example}{Example}

\begin{document}

\title{A Class of Two-Weight and Three-Weight Codes and Their Applications in Secret Sharing\thanks{The research of C. Ding was supported by the Hong Kong Research Grants Council, under Project No. 16301114.}}

\author{Kelan Ding\thanks{K. Ding is with the State Key Laboratory of Information Security, the Institute of Information Engineering, The Chinese Academy of Sciences, Beijing, China. Email: dingkelan@iie.ac.cn} and 
Cunsheng Ding\thanks{C. Ding is with the Department of Computer Science and Engineering, 
The Hong Kong University of Science and Technology, Clear Water Bay, Kowloon, Hong Kong. Email: cding@ust.hk} 
}

\date{\today}
\maketitle

\begin{abstract}
In this paper, a class of two-weight and three-weight linear codes over $\gf(p)$ is constructed, and their 
application in secret sharing is investigated. Some of the linear 
codes obtained are optimal in the sense that they meet certain bounds on linear codes. These codes have 
applications also in authentication codes, association schemes, and strongly regular graphs, 
in addition to their applications in consumer electronics, communication and data storage systems.      
\end{abstract}

\begin{keywords}
Association schemes, authentication codes, linear codes, secret sharing schemes, strongly regular graphs.   
\end{keywords}

\section{Introduction}\label{sec-intro} 

Throughout this paper, let $p$ be an odd prime and let $q=p^m$ for some positive integer $m$. 
An $[n,\, k,\,d]$ code $\C$ over $\gf(p)$ is a $k$-dimensional subspace of $\gf(p)^n$ with minimum 
(Hamming) distance $d$.  Let $A_i$ denote the number of codewords with Hamming weight $i$ in a code
$\C$ of length $n$. The {\em weight enumerator} of $\C$ is defined by
$
1+A_1z+A_2z^2+ \cdots + A_nz^n.
$ 
The {\em weight distribution} $(1,A_1,\ldots,A_n)$ is an important research topic in coding theory, 
as it contains crucial information as to estimate the error correcting capability and the probability of
error detection and correction with respect to some algorithms. 
A code $\C$ is said to be a $t$-weight code  if the number of nonzero
$A_i$ in the sequence $(A_1, A_2, \cdots, A_n)$ is equal to $t$.

Let $D=\{d_1, \,d_2, \,\ldots, \,d_n\} \subseteq \gf(q)$.
Let $\tr$ denote the trace function from $\gf(q)$ onto $\gf(p)$ throughout 
this paper. We define a linear code of 
length $n$ over $\gf(p)$ by 
\begin{eqnarray}\label{eqn-maincode} 
\C_{D}=\{(\tr(xd_1), \tr(xd_2), \ldots, \tr(xd_n)): x \in \gf(q)\},   
\end{eqnarray}  
and call $D$ the \emph{defining set} of this code $\C_{D}$.

This construction is generic in the sense that many classes of known codes 
could be produced by selecting the defining set $D \subseteq \gf(q)$. This 
construction technique was employed in \cite{DLN} and \cite{DN07} for 
obtaining linear codes with a few weights. 

The objective of this paper is to construct a class of linear codes over $\gf(p)$ with two and three 
nonzero weights using this generic construction method, and investigate their application in secret 
sharing. Some of the linear 
codes obtained in this paper are optimal in the sense that they meet some bounds on linear codes. 
The linear codes with a few weights presented in this paper have applications also in authentication codes \cite{CX05}, association schemes \cite{CG84}, and strongly  
regular graphs \cite{CG84}, in addition to their applications in consumer electronics, communication and data storage systems.

\section{The linear codes with two and three weights}  

We only describe the codes and introduce their parameters in this section. The proofs of their parameters will be 
given in Section \ref{sec-proof}.   

In this paper, the defining set $D$ of the code $\C_D$ of (\ref{eqn-maincode})  is given by 
\begin{eqnarray}\label{eqn-defsetD}
D=\{x \in \gf(q)^*: \tr(x^2)=0\}. 
\end{eqnarray}

\begin{theorem}\label{thm-twothree1}
Let $m>1$ be odd, and let $D$ be defined in (\ref{eqn-defsetD}). Then the set $\C_D$ of (\ref{eqn-maincode}) is a 
$[p^{m-1}-1, m]$ code over $\gf(p)$ 
with the weight distribution in Table \ref{tab-twothree1},  where $A_w=0$ for all other weights $w$ not listed in the table.  
\end{theorem} 

\begin{table}[ht]
\begin{center} 
\caption{The weight distribution of the codes of Theorem \ref{thm-twothree1}}\label{tab-twothree1}
\begin{tabular}{|c|c|} \hline
Weight $w$ &  Multiplicity $A_w$  \\ \hline  
$0$          &  $1$ \\ \hline 
$(p-1)\left(p^{m-2}  -  p^{\frac{m-3}{2}}\right)$  & 
$\frac{p-1}{2}\left(p^{m-1}  +  p^{\frac{m-1}{2}}\right)$\\ \hline 
$(p-1)p^{m-2}$  & $p^{m-1}-1$ \\ \hline 
$(p-1)\left(p^{m-2} + p^{\frac{m-3}{2}}\right)$  & 
$\frac{p-1}{2}\left(p^{m-1}  -  p^{\frac{m-1}{2}}\right)$\\ \hline 
\end{tabular}
\end{center} 
\end{table} 


\begin{example} 
Let $(p,m)=(3,5)$. Then the code $\C_D$ has parameters $[80, 5, 48]$ and weight enumerator 
$1+90z^{48}+80z^{54}+72z^{60}$. 
\end{example}

\begin{theorem}\label{thm-twothree2}
Let $m \geq 2$ be even, and let $D$ be defined in (\ref{eqn-defsetD}).  Then the code $\C_D$ over $\gf(p)$ 
of (\ref{eqn-maincode}) has parameters 
$$\left[p^{m-1}- (-1)^{(\frac{p-1}{2})^2 \frac{m}{2}} (p-1)p^{\frac{m-2}{2}} -1, m\right]$$ 
and  the weight distribution in Table \ref{tab-twothree2},  where $A_w=0$ for all other weights $w$ not listed in the table.  
\end{theorem} 

\begin{table}[ht]
\begin{center} 
\caption{The weight distribution of the codes of Theorem \ref{thm-twothree2}}\label{tab-twothree2}
\begin{tabular}{|c|c|} \hline
Weight $w$ &  Multiplicity $A_w$  \\ \hline  
$0$          &  $1$ \\ \hline 
$(p-1)p^{m-2}$  & $p^{m-1}- (-1)^{(\frac{p-1}{2})^2 \frac{m}{2}} (p-1)p^{\frac{m-2}{2}} -1$ \\ \hline 
$(p-1)\left(p^{m-2}-(-1)^{(\frac{p-1}{2})^2 \frac{m}{2}}p^{\frac{m-2}{2}}\right)$  & 
$(p-1)\left(p^{m-1}+(-1)^{(\frac{p-1}{2})^2 \frac{m}{2}}p^{\frac{m-2}{2}}\right)$\\ \hline 
\end{tabular}
\end{center} 
\end{table} 


\begin{example} 
Let $(p,m)=(5,4)$. Then the code $\C_D$ has parameters $[104, 4, 80]$ and weight enumerator 
$1+520z^{80}+104z^{100}$. The best linear code of length $104$ and dimension $4$ over $\gf(5)$ 
has minimum weight $81$. 
\end{example}

It is observed that the weights in the code $\C_D$ have a common divisor $p-1$. This indicates that the 
code $\C_D$ may be punctured into a shorter one whose weight distribution can be easily derived from that of 
the original code $\C_D$. This is indeed true and can be done as follows.  

Note that $\tr(ax^2)=0$ for all $a \in \gf(p)$ if $\tr(x^2)=0$. Hence, the set $D$ of (\ref{eqn-defsetD}) 
can be expressed as 
\begin{eqnarray}\label{eqn-dfsetbarD}
D=(\gf(p)^*) \bar{D}=\{ab: a \in \gf(p)^* \mbox{ and } b \in \bar{D}\}, 
\end{eqnarray}
where $d_i/d_j \not\in \gf(p)^*$ for every pair of distinct elements $d_i$ and $d_j$ in $\bar{D}$. Then the 
code $\C_{\bar{D}}$ is a punctured version of $\C_D$ whose parameters are given in the following two 
corollaries.  

\begin{corollary}\label{thm-twothree3}
Let $m>1$ be odd, and let $\bar{D}$ be defined in (\ref{eqn-dfsetbarD}). Then the set $\C_{\bar{D}}$ of (\ref{eqn-maincode}) is 
a $[(p^{m-1}-1)/(p-1), m]$ code over $\gf(p)$ 
with the weight distribution in Table \ref{tab-twothree3}, where $A_w=0$ for all other weights $w$ not listed in the table.   
\end{corollary} 

\begin{table}[ht]
\begin{center} 
\caption{The weight distribution of the codes of Corollary \ref{thm-twothree3}}\label{tab-twothree3}
\begin{tabular}{|c|c|} \hline
Weight $w$ &  Multiplicity $A_w$  \\ \hline  
$0$          &  $1$ \\ \hline 
$p^{m-2} -  p^{\frac{m-3}{2}}$  & 
$\frac{p-1}{2}\left(p^{m-1}  +   p^{\frac{m-1}{2}}\right)$\\ \hline 
$p^{m-2}$  & $p^{m-1}-1$ \\ \hline 
$p^{m-2} +  p^{\frac{m-3}{2}}$  & 
$\frac{p-1}{2}\left(p^{m-1} -  p^{\frac{m-1}{2}}\right)$\\ \hline 
\end{tabular}
\end{center} 
\end{table} 


\begin{example} 
Let $(p,m)=(3,5)$. Then the code $\C_{\bar{D}}$ has parameters $[40, 5, 24]$ and weight enumerator 
$1+90z^{24}+80z^{27}+72z^{30}$. This code is optimal in the sense that any ternary code of length 
$40$ and dimension $5$ cannot have minimum distance $25$ or more \cite{Eupen}.  
\end{example}

\begin{corollary}\label{thm-twothree4}
Let $m \geq 2$ be even, and let $\bar{D}$ be defined in (\ref{eqn-dfsetbarD}).  Then the code $\C_{\bar{D}}$ 
over $\gf(p)$ of (\ref{eqn-maincode}) has parameters  
$$\left[\frac{p^{m-1}-1}{p-1}- (-1)^{(\frac{p-1}{2})^2 \frac{m}{2}} p^{\frac{m-2}{2}}, m\right]$$ 
and the weight distribution in Table \ref{tab-twothree4},  where $A_w=0$ for all other weights $w$ not listed in the table. 
\end{corollary} 

\begin{table}[ht]
\begin{center} 
\caption{The weight distribution of the codes of Corollary \ref{thm-twothree4}}\label{tab-twothree4}
\begin{tabular}{|c|c|} \hline
Weight $w$ &  Multiplicity $A_w$  \\ \hline  
$0$          &  $1$ \\ \hline 
$p^{m-2}$  & $p^{m-1}- (-1)^{(\frac{p-1}{2})^2 \frac{m}{2}} (p-1)p^{\frac{m-2}{2}} -1$ \\ \hline 
$p^{m-2}-(-1)^{(\frac{p-1}{2})^2 \frac{m}{2}}p^{\frac{m-2}{2}}$  & 
$(p-1)\left(p^{m-1}+(-1)^{(\frac{p-1}{2})^2 \frac{m}{2}}p^{\frac{m-2}{2}}\right)$\\ \hline 
\end{tabular}
\end{center} 
\end{table} 


\begin{example} 
Let $(p,m)=(5,4)$. Then the code $\C_{\bar{D}}$ has parameters $[26, 4, 20]$ and weight enumerator 
$1+520z^{20}+104z^{25}$. This code is optimal due to the Griesmer bound. 
\end{example}

\section{The proofs of the main results}\label{sec-proof}

Our task of this section is to prove Theorems \ref{thm-twothree1} and \ref{thm-twothree2}, while 
Corollaries \ref{thm-twothree3} and \ref{thm-twothree4} follow directly from Theorems \ref{thm-twothree1} 
and \ref{thm-twothree2}, respectively. 

\subsection{Some auxiliary results}

To prove Theorems \ref{thm-twothree1} and \ref{thm-twothree2}, we need the help of a number of lemmas that are described 
and proved in the sequel. We start with group characters and Gauss sums. 

An {\em additive character} of $\gf(q)$ is a nonzero function $\chi$ 
from $\gf(q)$ to the set of nonzero complex numbers such that 
$\chi(x+y)=\chi(x) \chi(y)$ for any pair $(x, y) \in \gf(q)^2$. 
For each $b\in \gf(q)$, the function
\begin{eqnarray}\label{dfn-add}
\chi_b(c)=\epsilon_p^{\tr(bc)} \ \ \mbox{ for all }
c\in\gf(q) 
\end{eqnarray}
defines an additive character of $\gf(q)$, where $\epsilon_p=e^{2\pi \sqrt{-1}/p}$. When $b=0$,
$\chi_0(c)=1 \mbox{ for all } c\in\gf(q), 
$ 
and is called the {\em trivial additive character} of
$\gf(q)$. The character $\chi_1$ in (\ref{dfn-add}) is called the
{\em canonical additive character} of $\gf(q)$. 
It is known that every additive character of $\gf(q)$ can be 
written as $\chi_b(x)=\chi_1(bx)$ \cite[Theorem 5.7]{LN}. 

Since the multiplicative group $\gf(q)^*$ is cyclic, all the characters of the 
multiplicative group $\gf(q)^*$ are given by 
$$ 
\psi_j(\alpha^k)=e^{2\pi \sqrt{-1} jk/(q-1)}, \ \ k=0, 1, \cdots, q-2,  
$$
where $0 \le j \le q-2$ and $\alpha$ is a generator of $\gf(q)^*$. These $\psi_j$ are called {\em multiplicative characters} of $\gf(q)$, and form a group 
of order $q-1$ with identity element $\psi_0$. The character $\psi_{(q-1)/2}$ is called the 
{\em quadratic character} of $\gf(q)$, and is denoted by $\eta$ in this paper. We extend this quadratic character 
by letting $\eta(0)=0.$

The Gauss sum $G(\eta, \chi_1)$ over $\gf(q)$ is defined by 
\begin{eqnarray}
G(\eta, \chi_1)=\sum_{c \in \gf(q)^*} \eta(c) \chi_1(c) = \sum_{c \in \gf(q)} \eta(c) \chi_1(c)
\end{eqnarray} 
and 
the Gauss sum $G(\bar{\eta}, \bar{\chi}_1)$ over $\gf(p)$ is defined by 
\begin{eqnarray}
G(\bar{\eta}, \bar{\chi}_1)=\sum_{c \in \gf(p)^*} \bar{\eta}(c) \bar{\chi}_1(c) 
= \sum_{c \in \gf(p)} \bar{\eta}(c) \bar{\chi}_1(c), 
\end{eqnarray} 
where $\bar{\eta}$ and $\bar{\chi}_1$ are the quadratic and canonical additive characters of $\gf(p)$, 
respectively.  

The following lemma is proved in \cite[Theorem 5.15]{LN}. 

\begin{lemma}\label{lem-32A1}
With the symbols and notation above, we have   
$$ 
G(\eta, \chi_1)=(-1)^{m-1} \sqrt{-1}^{(\frac{p-1}{2})^2 m} \sqrt{q}
$$
and 
$$ 
G(\bar{\eta}, \bar{\chi}_1)= \sqrt{-1}^{(\frac{p-1}{2})^2 } \sqrt{p}. 
$$
\end{lemma} 

We will need the following lemma \cite[Theorem 5.33]{LN}. 

\begin{lemma}\label{lem-32A2}
Let $\chi$ be a nontrivial additive character of $\gf(q)$ with $q$ odd, and let 
$f(x)=a_2x^2+a_1x+a_0 \in \gf(q)[x]$ with $a_2 \ne 0$. Then 
$$ 
\sum_{c \in \gf(q)} \chi(f(c)) = \chi(a_0-a_1^2(4a_2)^{-1}) \eta(a_2) G(\eta, \chi).   
$$
\end{lemma} 

The conclusion of the following lemma is straightforward. For completeness, we provide a proof 
below. 

\begin{lemma}\label{lem-bothcharac}
If $m \ge 2$ is even, then $\eta(y)=1$ for each $y \in \gf(p)^*$. 
If $m$ is odd, then  $\eta(y)=\bar{\eta}(y)$ for each $y \in \gf(p)$. 
\end{lemma}  

\begin{proof}
Let $\alpha$ be a generator of $\gf(q)^*$. Notice that every $y \in \gf(p)^*$ can be expressed 
as $\alpha^{\frac{q-1}{p-1}j}$, where $0 \le j \le p-2$. We have 
\begin{eqnarray*}
\frac{q-1}{p-1} \bmod 2 = m \bmod 2. 
\end{eqnarray*} 
Hence, every element $y \in \gf(p)^*$ is a square in $\gf(q)$ when $m$ is an even positive integer, 
and $\eta(y)=\bar{\eta}(y)$ for each $y \in \gf(p)$ when $m$ is odd.  
This completes the proof. 
\end{proof}

Below we prove a few more auxiliary results before proving the main results of this paper. 

\begin{lemma}\label{lem-32B1} 
We have the following equality: 
\begin{eqnarray*}
\sum_{y \in \gf(p)^*} \sum_{x \in \gf(q)} \epsilon_p^{y\tr(x^2)} = 
\left\{ \begin{array}{ll}
0 & \mbox{ if $m$ odd,} \\
(-1)^{m-1} (-1)^{(\frac{p-1}{2})^2 \frac{m}{2}} (p-1)\sqrt{q} & \mbox{ if $m$ even.} 
\end{array}
\right. 
\end{eqnarray*}
\end{lemma} 

\begin{proof}
By Lemma \ref{lem-32A2}, we have 
\begin{eqnarray*}
\sum_{y \in \gf(p)^*} \sum_{x \in \gf(q)} \epsilon_p^{y\tr(x^2)} = 
G(\eta, \chi_1) \sum_{y \in \gf(p)^*} \eta(y).
\end{eqnarray*}
Using Lemma \ref{lem-bothcharac}, we obtain 
\begin{eqnarray*}
\sum_{y \in \gf(p)^*} \eta(y) = 
\left\{ \begin{array}{ll}
0 & \mbox{ if $m$ odd,} \\
p-1 & \mbox{ if $m$ even.} 
\end{array}
\right. 
\end{eqnarray*}
The desired conclusion then follows. 
\end{proof}

The next lemma will be employed later. 

\begin{lemma}\label{lem-32B2} 
For each $a \in \gf(p)$, let 
$$
n_a=|\{x \in \gf(q): \tr(x^2)=a\}|.  
$$
Then 
\begin{eqnarray*}
n_a=\left\{ \begin{array}{ll}
p^{m-1}                               & \mbox{ if $m$ odd and $a=0$,} \\
p^{m-1}- (-1)^{(\frac{p-1}{2})^2 \frac{m}{2}} (p-1)p^{\frac{m-2}{2}} & \mbox{ if $m$ even and $a=0$,} \\
p^{m-1}- \bar{\eta}(a) (-1)^{\frac{p-1}{2}} (-1)^{(\frac{p-1}{2})^2 (\frac{m+1}{2})} p^{\frac{m-1}{2}} & \mbox{ if $m$ odd and $a\ne 0$,} \\
p^{m-1}+ (-1)^{(\frac{p-1}{2})^2 \frac{m}{2}} p^{\frac{m-2}{2}} & \mbox{ if $m$ even and $a \ne 0$.} 
\end{array}
\right. 
\end{eqnarray*}
\end{lemma} 

\begin{proof}
It follows from Lemma \ref{lem-32A2} that 
\begin{eqnarray*}
n_a &=& \frac{1}{p} \sum_{x \in \gf(q)} \sum_{y \in \gf(p)} \epsilon_p^{y(\tr(x^2)-a)} \\
&=& p^{m-1}+ \frac{1}{p} \sum_{y \in \gf(p)^*} \epsilon_p^{ya}  \sum_{x \in \gf(q)}   \epsilon_p^{\tr(yx^2)} \\
&=& p^{m-1}+ \frac{1}{p} G(\eta, \chi_1) \sum_{y \in \gf(p)^*} \epsilon_p^{ya} \eta(y)  \\
&=& \left\{ \begin{array}{ll}
        p^{m-1}+ \frac{1}{p} G(\eta, \chi_1) \sum_{y \in \gf(p)^*} \eta(y) & \mbox{ if $a=0$} \\
        p^{m-1}+ \frac{1}{p} \eta(a)  G(\eta, \chi_1) \sum_{z \in \gf(p)^*} \epsilon_p^{z} \eta(z) & \mbox{ if $a\ne 0$} 
\end{array}
\right. \\ 
&=& \left\{ \begin{array}{ll}
        p^{m-1}                                                                                               & \mbox{ if $m$ odd and $a=0$,} \\
        p^{m-1}+ \frac{p-1}{p} G(\eta, \chi_1)  & \mbox{ if $m$ even and $a=0$,} \\        
        p^{m-1}+ \frac{\bar{\eta}(a)}{p}   G(\eta, \chi_1) G(\bar{\eta}, \bar{\chi}_1) & \mbox{ if $m$ odd and $a\ne 0$,} \\ 
        p^{m-1}- \frac{1}{p}  G(\eta, \chi_1)  & \mbox{ if $m$ even and $a\ne 0$,}         
\end{array}
\right. 
\end{eqnarray*}
where the first equality follows from the fact that $\sum_{y \in \gf(p)} \bar{\chi}_1(yz)=0$ for every 
$z \in \gf(p)^*$. 
The desired conclusion then follows from Lemma \ref{lem-32A1}. 
\end{proof}

The following result will play an important role in proving the main results of this paper. 

\begin{lemma}\label{lem-32B3}
Let $b \in \gf(q)^*$. Then 
\begin{eqnarray*}
\lefteqn{\sum_{y \in \gf(p)^*} \sum_{z \in \gf(p)^*} \sum_{x \in \gf(q)}  \epsilon_p^{\tr(yx^2+bzx)} } \\ 
&=&  
\left\{ \begin{array}{ll}
0                                                                  & \mbox{if $m$ odd and $\tr(b^2)=0$,} \\
-\bar{\eta}(\tr(b^2)) (-1)^{(\frac{p-1}{2})^2 (\frac{m+1}{2})} (p-1) p^{\frac{m+1}{2}} & \mbox{if $m$ odd and $\tr(b^2)\ne 0$,} \\
- (-1)^{(\frac{p-1}{2})^2 \frac{m}{2}} (p-1)^2 p^{\frac{m}{2}} & \mbox{if $m$ even and $\tr(b^2) = 0$,} \\
 (-1)^{(\frac{p-1}{2})^2 \frac{m}{2}} (p-1) p^{\frac{m}{2}} & \mbox{if $m$ even and $\tr(b^2)\ne 0$.} \\
\end{array}
\right. 
\end{eqnarray*}
\end{lemma} 

\begin{proof}
It follows from Lemmas \ref{lem-32A2} and \ref{lem-bothcharac} that 
\begin{eqnarray*}
\lefteqn{\sum_{y \in \gf(p)^*} \sum_{z \in \gf(p)^*} \sum_{x \in \gf(q)}  \epsilon_p^{\tr(yx^2+bzx)} } \\
&=& G(\eta, \chi_1) \sum_{y \in \gf(p)^*} \sum_{z \in \gf(p)^*} \chi_1\left(-\frac{b^2z^2}{4y}\right) \eta(y) \\
&=& G(\eta, \chi_1) \sum_{y_1 \in \gf(p)^*} \sum_{z \in \gf(p)^*} \chi_1\left(-b^2z^2y_1\right) \eta\left(\frac{1}{4y_1}\right) \\
&=& G(\eta, \chi_1) \sum_{y_1 \in \gf(p)^*} \sum_{z \in \gf(p)^*} \chi_1\left(-b^2z^2y_1\right) \eta\left(\frac{y_1}{(2y_1)^2}\right) \\
&=& G(\eta, \chi_1) \sum_{y \in \gf(p)^*} \sum_{z \in \gf(p)^*} \chi_1\left(-b^2z^2y\right) \eta\left(y\right) \\
&=& G(\eta, \chi_1) \sum_{y \in \gf(p)^*} \sum_{z \in \gf(p)^*} \epsilon_p^{-z^2\tr(b^2)y} \eta\left(y\right) \\
&=& \left\{ \begin{array}{ll}
G(\eta, \chi_1) \sum_{z \in \gf(p)^*} \sum_{y \in \gf(p)^*} \eta\left(y\right) & \mbox{ if $\tr(b^2)=0$} \\
 G(\eta, \chi_1) \sum_{z \in \gf(p)^*} \sum_{y \in \gf(p)^*} \epsilon_p^{-z^2\tr(b^2)y} \eta(-z^2\tr(b^2)y) \eta(-\tr(b^2)) 
 & \mbox{ if $\tr(b^2) \ne 0$} 
\end{array}
\right. \\
&=& \left\{ \begin{array}{ll}
G(\eta, \chi_1) (p-1) \sum_{y \in \gf(p)^*} \eta\left(y\right) & \mbox{ if $\tr(b^2)=0$} \\
 G(\eta, \chi_1) \eta(-\tr(b^2)) (p-1) \sum_{y \in \gf(p)^*} \epsilon_p^{y} \eta(y)  
 & \mbox{ if $\tr(b^2) \ne 0$} 
\end{array}
\right. \\
&=& \left\{ \begin{array}{ll}
0                              & \mbox{ if $m$ odd and $\tr(b^2)=0$,} \\
 G(\eta, \chi_1)  G(\bar{\eta}, \bar{\chi}_1) \eta(-\tr(b^2)) (p-1)  
 & \mbox{ if $m$ odd and $\tr(b^2) \ne 0$,} \\ 
G(\eta, \chi_1) (p-1)^2  & \mbox{ if $m$ even and $\tr(b^2)=0$,} \\
-G(\eta, \chi_1) (p-1)  & \mbox{ if $m$ even and $\tr(b^2) \ne 0$.} 
\end{array}
\right. 
\end{eqnarray*}
The desired conclusions then follow from Lemmas \ref{lem-32A1} and \ref{lem-bothcharac}. 
\end{proof}

The last auxiliary result we need is the following. 

\begin{lemma}\label{lem-32B4} 
For any $b \in \gf(q)^*$ and any $a \in \gf(p)$, let 
$$ 
N(b)=|\{x \in \gf(q): \tr(x^2)=0 \mbox{ and } \tr(bx)=0\}|. 
$$
Then 
\begin{eqnarray*}
N(b)= \left\{ \begin{array}{ll}
p^{m-2}                    &   \mbox{if $m$ odd and $\tr(b^2)=0$,} \\
p^{m-2}-\bar{\eta}(\tr(b^2)) (-1)^{\left( \frac{p-1}{2}\right)^2 (\frac{m+1}{2})} (p-1)p^{\frac{m-3}{2}}          
              &   \mbox{if $m$ odd and $\tr(b^2) \ne 0$,} \\
p^{m-2}- (-1)^{\left( \frac{p-1}{2}\right)^2 \frac{m}{2}} (p-1)p^{\frac{m-2}{2}}          
              &   \mbox{if $m$ even and $\tr(b^2) = 0$,} \\   
p^{m-2}                    &   \mbox{if $m$ even and $\tr(b^2)\ne 0$.}                          
\end{array}
\right. 
\end{eqnarray*}
\end{lemma}

\begin{proof}
By definition, we have 
\begin{eqnarray*}
N(b)  
&=& p^{-2} \sum_{x \in \gf(q)} \left( \sum_{y \in \gf(p)} \epsilon_p^{y\tr(x^2)} \right)  
                                                   \left( \sum_{z \in \gf(p)} \epsilon_p^{z\tr(bx)} \right) \\
&=& p^{-2}  \sum_{z \in \gf(p)^*} \sum_{x \in \gf(q)}  \epsilon_p^{\tr(bzx)}  + 
        p^{-2}  \sum_{y \in \gf(p)^*} \sum_{x \in \gf(q)}  \epsilon_p^{\tr(yx^2)} + \\
& & p^{-2}   \sum_{y \in \gf(p)^*}  \sum_{z \in \gf(p)^*} \sum_{x \in \gf(q)}  \epsilon_p^{\tr(yx^2+bzx)} + p^{m-2}.                                        
\end{eqnarray*}
Note that 
$$ 
\sum_{z \in \gf(p)^*} \sum_{x \in \gf(q)}  \epsilon_p^{\tr(bzx)} =0. 
$$
The desired conclusions then follow from Lemmas \ref{lem-32B1} and \ref{lem-32B3}. 
\end{proof}

\subsection{The proof of Theorems \ref{thm-twothree1} and \ref{thm-twothree2}}

It follows from Lemma \ref{lem-32B2} that the length $n$ of the code $\C_D$ is given by 
\begin{eqnarray*}
n = |D|=n_0-1=\left\{ \begin{array}{ll}
p^{m-1}-1                               & \mbox{if $m$ odd,} \\
p^{m-1}-1- (-1)^{(\frac{p-1}{2})^2 \frac{m}{2}} (p-1)p^{\frac{m-2}{2}} & \mbox{if $m$ even.} 
\end{array}
\right. 
\end{eqnarray*} 

For each $b \in \gf(q)^*$, define 
\begin{eqnarray}\label{eqn-mcodeword}
\bc_{b}=(\tr(bd_1), \,\tr(bd_2), \,\ldots, \,\tr(bd_n)),   
\end{eqnarray} 
where $d_1, d_2, \ldots, d_n$ are the elements of $D$. 
The Hamming weight $\wt(\bc_b)$ of $\bc_b$ is $n_0-N(b)$, where $n_0$ and $N(b)$ 
were defined before. 

When $m$ is odd, it follows from Lemmas \ref{lem-32B2} and \ref{lem-32B4} that 
\begin{eqnarray*}
\wt(\bc_b)=n_0-N(b)= \left\{ \begin{array}{ll}
(p-1)p^{m-2}                                            & \mbox{if $\tr(b^2)=0$,} \\
(p-1)\left(p^{m-2}+\bar{\eta}(\tr(b^2)) (-1)^{\left(\frac{p-1}{2} \right)^2 (\frac{m+1}{2})} p^{\frac{m-3}{2}} \right) 
                                                                  & \mbox{if $\tr(b^2) \ne 0$.} 
\end{array}
\right. 
\end{eqnarray*}
The desired conclusions of Theorem \ref{thm-twothree1} then follow from Lemma \ref{lem-32B2} and the fact that 
$\wt(\bc_b)>0$ for each $b \in \gf(q)^*$. 

When $m$ is even, it follows from Lemmas \ref{lem-32B2} and \ref{lem-32B4} that 
\begin{eqnarray*}
\wt(\bc_b)=n_0-N(b)= \left\{ \begin{array}{ll}
(p-1)p^{m-2}                                            & \mbox{if $\tr(b^2)=0$,} \\
(p-1)\left(p^{m-2} - (-1)^{\left(\frac{p-1}{2} \right)^2 \frac{m}{2}} p^{\frac{m-2}{2}} \right) & \mbox{if $\tr(b^2) \ne 0$.} 
\end{array}
\right. 
\end{eqnarray*}
The desired conclusions of Theorem \ref{thm-twothree2} then follow from Lemma \ref{lem-32B2} and the fact that 
$\wt(\bc_b)>0$ for each $b \in \gf(q)^*$.

\section{A generalization of the construction} 

Let $f$ be a function from a finite abelian group $(A, +)$ to a finite abelian
group $(B, +)$. 
A robust measure of nonlinearity of $f$ is defined by
\begin{eqnarray*}\label{eqn-nonl02}
P_f=
\max_{0 \neq a \in A} \max_{b \in B} \frac{|\{x \in A : f(x+a)-f(x)=b \}|}{|A|}.   
\end{eqnarray*} 
The smaller the value of $P_f$, the higher the corresponding nonlinearity of $f$. 

It is easily seen that $ P_f \geq \frac{1}{|B|}$ \cite{CD04}. A function $f:
A \to B$ has {\em perfect nonlinearity} if $P_f = \frac{1}{|B|}$. A perfect 
nonlinear function from a finite abelian group to a finite abelian group of 
the same order is called a {\em planar function} in finite geometry. 
Planar functions were introduced by Dembowski and Ostrom in 1968 for the 
construction of affine planes \cite{DO68}. 
We refer to Carlet and Ding \cite{CD04} for a survey of highly nonlinear 
functions, Coulter and Matthews \cite{Coult} and Ding and Yuan \cite{DY04} 
for information about planar functions. 

Some known planar functions from $\gf(q)$ to $\gf(q)$ are the
following \cite{CD04,Coult}:
\begin{itemize}
\item $f(x)=x^2$.

\item $f(x)=x^{p^k+1}$, where $m/\gcd(m,k)$ is odd (Dembowski and Ostrom \cite{DO68}).

\item $f(x)=x^{\frac{3^k+1}{2}}$, where $p=3$, $k$ is odd, and
$\gcd(m,k)=1$ (Coulter and Matthews \cite{Coult}).

\item $f_u(x)= x^{10}-u x^6 -u^2 x^2$, where $p=3$ and $m$ is
odd (Coulter and Matthews \cite{Coult} for the case $u=-1$, Ding and Yuan \cite{DY04} 
for the general case).
\end{itemize} 
Note that planar functions over $\gf(p^m)$ exist for any pair $(p, m)$ with $p$ being 
an odd prime number. 

The construction of the linear code $\C_D$ of this paper can be generalized as follows. Let $f$ be a planar 
function from $\gf(q)$ to $\gf(q)$ such that 
\begin{itemize}
\item $f(0)=0$; 
\item $f(x)=f(-x)$ for all $x \in \gf(q)$; and 
\item $f(ax)=a^hf(x)$ for all $a \in \gf(p)$ and $x \in \gf(q)$, where $h$ is some constant. 
\end{itemize}
Then the set 
$$ 
D_f:=\{x \in \gf(q)^*: \tr(f(x))=0\} \subset \gf(q) 
$$
defines a linear code $\C_{D_f}$ over $\gf(p)$. The code $\C_{D_f}$ may have the same parameters as the code 
$\C_D$ of this paper. Magma confirms that this is true for all the four classes of planar functions listed above. But it 
is open whether $\C_{D_f}$ and $\C_{D}$ have the same parameters and weight distribution for any planar function 
$f$ satisfying the three conditions above. It would 
be nice if this open problem can be settled. 

We remark that this construction of linear codes with planar functions here 
is different from the one in \cite{CDY05}, as the lengths and dimensions of the codes in the two constructions are different.

\section{Applications of the linear codes in secret sharing schemes}

In this section, we describe and analyse the secret sharing schemes from  some of the codes presented in this paper. 

\subsection{Secret sharing schemes}

A secret sharing scheme consists of 
\begin{itemize}
\item a dealer, and a group $\cP=\{P_1, P_2, \cdots, P_\ell \}$ of $\ell$ participants;  
\item a secret space $\cS$; 
\item $\ell$ share spaces $\cS_1$, $\cS_2$, $\cdots$, $\cS_\ell$; 
\item a share computing procedure; and 
\item a secret recovering procedure. 
\end{itemize} 

The dealer will choose a secret $s$ from the secret space $\cS$, and will employ the sharing computing procedure 
to compute a share of the secret $s$ for each participant $P_i$, and then give the share to $P_i$. 
The share computed for $P_i$ belongs to the share space $\cS_i$.   
When a subset of the participants comes together with their shares, they may be able to recover the secret $s$ from their 
shares with the secret recovering procedure. The secret $s$ and the sharing computing function are known only to the dealer, 
while the secret recovering procedure is known to all the participants. 

By an {\em access set} we mean a group of participants who can determine the secret from their shares. 
The {\em access structure} of a secret sharing scheme is defined to be the set of all access sets. 
A {\em minimal access set} is a group of participants who can recover the 
      secret with their shares, but any of its proper subgroups cannot do so. 
A secret sharing scheme is said to have the {\em monotone access structure},
if any superset of any access set is also an access set.  
In a secret sharing scheme with the monotone access structure, the access structure is totally characterized 
by its minimal access sets by definition. In this section, we deal with secret sharing schemes only with the 
monotone access structure.   

Secret sharing schemes have applications in banking systems, cryptographic protocols, electronic voting 
systems, and the control of nuclear weapons. In 1979,  Shamir and Blakley documented the first secret 
sharing schemes in the literature \cite{Blakley,Shamir}.  

\subsection{The covering problem of linear codes} 

In order to describe the secret sharing scheme of a linear code, we need to introduce the covering problem 
of linear codes. 

The {\em support} of a vector $\bc=(c_0, \ldots, c_{n-1}) \in \gf(p)^n$ 
is defined as 
$$
\{0 \leq i \leq n-1: c_i \neq 0\}.
$$
We say that a vector $\bx$ covers a vector $\by$ if the support of $\bx$ contains
that of $\by$ as a proper subset. 

A {\em minimal codeword} of a linear code $\C$ is a nonzero codeword that does not cover 
any other nonzero codeword  of $\C$. The {\em covering problem} of a linear code is to determine all the 
minimal codewords of $\C$. This is a very hard problem in general, but can be solved for certain types of 
linear codes.

\subsection{A construction of secret sharing schemes from linear codes}

Any linear code over $\gf(p)$ can be employed to construct secret sharing schemes \cite{ADHK,CDY05,Mass93,YD06}. 
Given a linear code $\C$ over $\gf(p)$ with parameters $[n, k, d]$ and generator matrix $G=[\bg_0, \bg_1, \ldots, \bg_{n-1}]$, 
we use $d^\perp$  and $H=[\bh_0, \bh_1, \ldots, \bh_{n-1}]$ to denote the minimum distance and the generator matrix of 
its dual code $\C^\perp$. 

In the secret sharing scheme based on $\C$, the secret space and the share spaces all are $\gf(p)$, and the participants are 
denoted by $P_1, P_{2}, \cdots, P_{n-1}$. To compute shares for all the participants, The dealer chooses randomly a 
vector $\bu=(u_0,\ldots, u_{n-k-1})$ such that $s=\bu\bh_0$, which is the inner product of the two vectors. 
The dealer then treats $\bu$ as an information vector and 
computes the corresponding codeword $$\bt=(t_0,t_1,\ldots,t_{n-1})=\bu H.$$
He then gives $t_i$ to party $P_i$ as his/her share for each $i\geq 1$.

The secret recovering procedure is the following. 
 Note that $t_0=\bu\bh_0=s$. 
A set of shares $\{t_{i_1},t_{i_2},\ldots,t_{i_m}\}$ determines the secret $s$ iff $\bh_0$ is a linear combination of
 $\bh_{i_1},\ldots,\bh_{i_m}$. Suppose that 
 $$
 \bh_0= \sum_{j=1}^m x_j \bh_{i_j}.
 $$
 Then the secret $s$ is recovered by computing
 $$
 s=\sum_{j=1}^m x_j t_{i_j}.
 $$
Equivalently, we look for codewords $\bc$ of the code $\C$ with the shape
$$
 (1,0,\ldots,0,c_{i_1},0,\ldots,0,c_{i_m},0,\ldots,0)
$$ 

Hence, the minimal access sets of the secret sharing scheme based on $\C^\perp$
correspond to the minimal codewords in $\C$ having 1 as
their leftmost component. The other nonzero components correspond 
to the participants in the minimal access set. For example, if $(1, 2, 0, 0, 2)$ is a codeword of $\C$,
then  $\{P_1,P_4\}$ is a minimal access set.  
To obtain the access structure of the secret sharing scheme based on $\C^{\perp}$, we need to determine 
all minimal
codewords of $\C$. 

Note that the access structure of the secret sharing scheme based on $\C^\perp$ is independent of 
the choice of the generator matrix $H$ of $\C^\perp$. We therefore  say that the 
secret sharing scheme is based on $\C^\perp$ without mentioning the matrix $H$. We would remind 
the reader that a linear code gives a pair of secret sharing schemes. One is based on $\C$ and the other is 
based on $\C^\perp$. Below we consider only the latter due to symmetry. 

The access structure of the secret sharing scheme based on a linear code is very complex in general, 
but can be determined in certain special cases. 
The following theorem is proved in \cite{DY03,YD04}. 

\begin{theorem}\label{thm-DingYuan}
Let $\C$ be an $[n, k, d]$ code over $\gf(p)$, and let $G=[\bg_0, \bg_1, \cdots, \bg_{n-1}]$
be its generator matrix. Let $d^{\perp}$ denote the minimum distance of its 
dual code $\C^{\perp}$. If each nonzero codeword of $\C$ is minimal,
then in the secret sharing scheme based on $\C^{\perp}$, the total number of participants is $n-1$, 
and there are altogether
$p^{k-1}$ minimal access sets.
\begin{itemize}
\item When $d^{\perp}=2$, the access structure is as follows.

If $\bg_i$ is a multiple of $\bg_0$, $1 \leq i \leq n-1$,
      then participant $P_i$ must be in every minimal access set.

If $\bg_i$ is not a multiple of $\bg_0$, $1 \leq i \leq
      n-1$, then participant $P_i$ must be in $(p-1)p^{k-2}$ out of $p^{k-1}$
      minimal access sets.

\item When $d^{\perp} \geq 3$, for any fixed $1 \leq t \leq
      \min\{k-1, d^{\perp}-2\}$
      every group of $t$ participants is involved in $(p-1)^t p^{k-(t+1)}$
      out of $p^{k-1}$ minimal access sets.
\end{itemize}
\end{theorem}

When the conditions of Theorem \ref{thm-DingYuan} are satisfied, the secret sharing scheme based 
on the dual code $\C^\perp$ is interesting. In the case that $d^{\perp}=2$, some participants must 
be in every minimal access sets, and thus are dictators. Such a secret sharing scheme may be required 
in certain applications. In the case that $d^{\perp} \geq 3$, each participant plays the same role as 
he/she is involved in the same number of minimal access sets. Such a secret sharing scheme is said 
to be {\em democratic}, and may be needed in some other application scenarios.   

A question now is how to construct a linear code whose nonzero codewords all are minimal. The following 
lemma provides a guideline in this direction \cite{AsBa95,AsBa98}. 

\begin{lemma}\label{lem-Ash} 
Every nonzero codeword of a linear code $\C$ over $\gf(p)$ is minimal, provided that 
$$
\frac{w_{min}}{w_{max}}>\frac{p-1}{p},
$$
where $w_{max}$ and $w_{min}$ denote the maximum and minimum nonzero weights in $\C$, 
respectively. 
\end{lemma}

\subsection{The secret sharing schemes from the codes of this paper}

In this subsection, we consider the secret sharing schemes based on the dual codes $\C_D^\perp$ and 
$\C_{\bar{D}}^\perp$ of the codes $\C_D$ and $\C_{\bar{D}}$ presented in this paper. 

For the code $\C_D$ of Theorem \ref{thm-twothree1} and  the code $\C_{\bar{D}}$ of Corollary  
\ref{thm-twothree3}, we have 
\begin{eqnarray*}
\frac{w_{\min}}{w_{\max}} = \frac{p^{m-2}-p^{\frac{m-3}{2}}}{p^{m-2}+p^{\frac{m-3}{2}}} > \frac{p-1}{p}
\end{eqnarray*}
if $m \geq 5$. 

Let $m \equiv 0 \pmod{4}$ or $m \equiv 0 \pmod{2}$ and $p \equiv 1 \pmod{4}$. 
Then for the code $\C_D$ of Theorem \ref{thm-twothree2} and  the code $\C_{\bar{D}}$ of Corollary  
\ref{thm-twothree4}, we have 
\begin{eqnarray*}
\frac{w_{\min}}{w_{\max}} = \frac{p^{m-2}-p^{\frac{m-2}{2}}}{p^{m-2}} > \frac{p-1}{p}
\end{eqnarray*}
if $m \geq 4$. 

Let $m \equiv 2 \pmod{4}$ and $p \equiv 1 \pmod{4}$. 
Then for the code $\C_D$ of Theorem \ref{thm-twothree2} and  the code $\C_{\bar{D}}$ of Corollary  
\ref{thm-twothree4}, we have 
\begin{eqnarray*}
\frac{w_{\min}}{w_{\max}} = \frac{p^{m-2}}{p^{m-2}+p^{\frac{m-2}{2}}} > \frac{p-1}{p}
\end{eqnarray*}
if $m \geq 6$. 

It then follows from Lemma \ref{lem-Ash} that all the nonzero codewords of  $\C_{D}$ 
and $\C_{\bar{D}}$ are minimal if $m \ge 6$. Hence, the secret sharing schemes based on 
the dual codes $\C_D^\perp$ and $\C_{\bar{D}}^\perp$ have the nice access structures described 
in Theorem \ref{thm-DingYuan}. 

As an example, we describe the access structure of the secret sharing scheme based on the dual code 
$\C_{\bar{D}}^\perp$ of the code  $\C_{\bar{D}}$ of Corollary \ref{thm-twothree3} as follows. 

\begin{corollary}\label{cor-secrets} 
Let $m \geq 5$. In the secret sharing scheme based on the dual code 
$\C_{\bar{D}}^\perp$ of the code  $\C_{\bar{D}}$ of Corollary \ref{thm-twothree3}, the 
total number of participants is $p^{m-2}$, and the total number of minimal access sets is 
$p^{m-1}$. Every participant is a member of exactly $(p-1)p^{m-2}$ minimal access sets.   
\end{corollary}

\begin{proof} 
As proved above, every nonzero codeword of $\C_{\bar{D}}$ is minimal as $m \geq 5$. 
It can be easily proved that $d^\perp \geq 3$. The desired conclusions then follow from Theorem 
\ref{thm-DingYuan}. 
\end{proof}

As an example of Corollary \ref{cor-secrets}, we have the following. 
\begin{example}\label{exam-secrets} 
Let $m=5$ and $p=5$. In the secret sharing scheme based on the dual code 
$\C_{\bar{D}}^\perp$ of the code  $\C_{\bar{D}}$ of Corollary \ref{thm-twothree3}, the 
total number of participants is $125$, and the total number of minimal access sets is 
$625$. Every participant is a member of exactly $500$ minimal access sets.    
\end{example}

In the secret sharing scheme of Example \ref{exam-secrets}, the secret space is $\gf(5)$, 
which is too small. However, it can still be employed for sharing a secret of any size. This is 
done as follows. One can have $\gf(5^h)$ as the extended secret space, where $h$ could be as large 
as one wants (e.g., $h=60$). Then any secret can be encoded as a sequence 
$$ 
s=s_1s_2 \ldots s_h   
$$
using an encoding scheme, where each $s_i \in \gf(5)$. 
Then the secret $s$ can be shared by the 125 participants symbol by symbol  
with the secret sharing scheme of Example \ref{exam-secrets}. Hence, the share for each 
participant will be a sequence of elements of $\gf(5)$ with length $h$. When a group of 
participants come together with their shares, the elements $s_i$ in the secret $s$ will be 
recovered one by one using the corresponding elements in their shares. 

Finally, we mention that the secret sharing scheme based on the dual code 
$\C_{\bar{D}}^\perp$ of the code  $\C_{\bar{D}}$ of Corollary \ref{thm-twothree4} 
has a similar access structure as the one described in Corollary \ref{cor-secrets}. 
For the linear codes of Theorems \ref{thm-twothree1} and \ref{thm-twothree2}, 
their dual codes have minimum distance $2$. Hence, the secret sharing scheme based 
on the dual code $\C_{\bar{D}}^\perp$ of the code  $\C_{\bar{D}}$ in 
Theorems \ref{thm-twothree1} and \ref{thm-twothree2} have dictators in 
the whole group of participants. Their access structure is given in the first 
case of Theorem \ref{thm-DingYuan}.

\section{Concluding remarks} 

Calderbank and Kantor surveyed two-weight codes in \cite{CK85}. There is a recent survey on three-weight 
cyclic codes \cite{DLLZ}.  Some interesting two-weight and three-weight codes were presented in \cite{CG84}, \cite{CW84}, 
\cite{Choi}, \cite{FL07}, \cite{LiYueLi}, \cite{LiYueLi2}, \cite{RP10}, \cite{Xia}, and \cite{ZD13}. 
The length of the two-weight and three-weight codes in the literature usually divides $p^m-1$, while 
that of the codes presented in this paper does not have this property.  
We did not find the parameters of the two-weight and three-weight codes of this paper in the literature.  

The two-weight codes $\C_D$ of this paper give automatically strongly regular graphs having new parameters 
with the connection described in \cite{CK85}, 
and the three-weight codes $\C_D$ of this paper may yield association schemes having new parameters 
with the framework introduced in \cite{CG84}. 
The linear codes of this paper can be employed to construct authentication codes having new parameters 
via the framework in \cite{DHKW,CX05}. For this application, we need to know not only 
the weight distribution of the linear codes, but also the distribution of each element of $\gf(p)$ in 
each codeword of the linear code. This is called the {\em complete weight distribution} of a code. 
Another advantage of the linear codes in this paper is that their complete weight distribution can be 
settled with the help of Gaussian sums. In the literature the complete weight distribution of only a 
few classes of linear codes is known.  

Compared with other two-weight and three-weight codes, the construction method of the codes in this paper 
is very simple and is defined by the simple function $\tr(x^2)$. This makes the analysis of the linear codes much 
easier.

\section*{Acknowledgements} 
The authors are very grateful to the reviewers and the Associate Editor, Dr. Yongyi Mao, for their comments and suggestions
that improved the presentation and quality of this paper.


\begin{thebibliography}{99} 

\bibitem{ADHK}  R. Anderson, C. Ding, T. Helleseth and T. Kl\o ve, ``How to build robust shared control systems," 
\emph{Designs, Codes and Cryptography}, vol. 15, no. 2, pp 111--124, 1998. 

\bibitem{AsBa95} A. Ashikhmin, A. Barg, G. Cohen and L. Huguet,
          ``Variations on minimal codewords in linear codes, 
          in {\em Proc. of AAECC 1995,} pp. 96--105,  LNCS 948, 
          Springer-Verlag, 1995. 

 
\bibitem{AsBa98} A. Ashikhmin and A. Barg, ``Minimal vectors in 
         linear codes,'' {\em IEEE Trans. Inform. Theory,} 
         vol. 44, no. 5,  pp. 2010--2017, 1998.


\bibitem{Blakley} G. R. Blakley, ``Safeguarding cryptographic keys", in: 
{\em Proceedings of the National Computer Conference}, vol. 48, pp. 313--317, 1979.

\bibitem{CG84} A. R. Calderbank and J. M. Goethals, ``Three-weight codes and association schemes," 
{\em Philips J. Res.}, vol. 39, pp. 143--152, 1984. 

\bibitem{CK85} 
A. R. Calderbank and W. M. Kantor, ``The geometry of two-weight codes," {\em Bull. London Math. Soc.}, 
vol. 18, pp. 97--122,  1986.    

\bibitem{CD04} C. Carlet and C. Ding, ``Highly nonlinear mappings,'' {\em J. Complexity,} 
vol. 20, no. 2, pp. 205--244, 2004.  

\bibitem{CDY05} C. Carlet, C. Ding and J. Yuan, ``Linear codes from perfect nonlinear
mappings and their secret sharing schemes," {\em IEEE Trans. Inform. Theory,} 
vol. 51, no. 6, pp. 2089--2102, 2005.

\bibitem{Choi}  S.-T. Choi, J.-Y. Kim, J.-S. No and H. Chung, ``Weight distribution of some cyclic codes,"
in: {\em Proc. of the 2012 International Symposium on Information Theory}, pp. 2911--2913, IEEE Press, 2012.

\bibitem{Coult} R. S. Coulter and R. W. Matthews, ``Planar functions and planes of Lenz-Barlotti
class II,'' {\em Designs, Codes and Cryptography,} vol. 10, pp. 167--184, 1997. 

\bibitem{CW84} B. Courteau and J. Wolfmann, ``On triple-sum-sets and two or three weight codes," 
\emph{Discrete Mathematics}, vol. 50, pp. 179--191, 1984. 

\bibitem{DO68} P. Dembowski and T. G. Ostrom, ``Planes of order $n$ with collineation 
       groups of order $n^2$,'' {\em Math. Z.}, vol. 193, pp. 239--258, 1968. 
       
\bibitem{DHKW} C. Ding, T. Helleseth, T. Kl{\o}ve and X. Wang, ``A general construction of authentication codes," 
{\em IEEE
Trans. Inform. Theory,} vol. 53, no. 6, pp. 2229--2235, 2007.       
       
\bibitem{DLLZ} C. Ding, C. Li, N. Li and Z. Zhou, ``Three-weight cyclic codes and their weight distributions," 
Preprint, 2014.        

\bibitem{DLN} C. Ding, J. Luo and H. Niederreiter, ``Two weight codes punctured from irreducible 
cyclic codes," in: 
\emph{Proc. of the First International Workshop on Coding Theory and Cryptography}, 
pp. 119--124. Singapore, World Scientific, 2008.

\bibitem{DN07} C. Ding and H. Niederreiter, ``Cyclotomic linear codes of order 3", 
\emph{IEEE Trans. Inform. Theory}, vol. 53, no. 6, pp. 2274--2277, 2007. 

\bibitem{CX05} C. Ding and X. Wang, ``A coding theory construction of new systematic authentication codes," 
\emph{Theoretical Computer Science}, vol. 330, pp. 81--99, 2005. 

\bibitem{DY03} C. Ding and J. Yuan, ``Covering and secret sharing with
         linear codes,'' in: {\em Discrete Mathematics and Theoretical 
         Computer Science,} pp. 11--25, LNCS 2731, Springer Verlag,
         2003.   

\bibitem{DY04} C. Ding and J. Yuan, ``A  family of skew Paley-Hadamard difference sets,'' 
{\em J. of Combinatorial Theory A}, vol. 113, no. 7, pp. 1219--1592,  2006.    

\bibitem{Eupen} M. van Eupen, ``Some new results for ternary linear codes of dimension $5$ and $6$, 
{\em IEEE Trans. Inform. Theory}, vol. 41, no. 6,  pp. 2048--2051, 1995.


\bibitem{FL07} K. Feng and J. Luo, ``Value distribution of exponential sums from
perfect nonlinear functions and their applications," {\em IEEE Trans. Inform.
Theory}, vol. 53, no. 9, pp. 3035--3041, 2007.

\bibitem{LiYueLi} C. Li, Q. Yue and F. Li, ``Weight distributions of cyclic codes with respect to pairwise 
coprime order elements," \emph{Finite Fields and Their Applications,} vol. 28, pp. 94--114, 2014. 

\bibitem{LiYueLi2} C. Li, Q. Yue and F. Li, ``Hamming weights of the duals of cyclic codes with two zeros," 
\emph{IEEE Trans. Inform. Theory,} vol. 60, no. 7, pp. 3895--3902, 2014. 

\bibitem{LN} R. Lidl and H. Niederreiter, {\em Finite Fields,} Cambridge: Cambridge University Press,
1997. 

\bibitem{Mass93} J. L. Massey, ``Minimal codewords and secret sharing,'' 
         in: {\em Proc. 6th Joint Swedish-Russian Workshop on Information 
         Theory,}  pp. 276--279, 1993.

\bibitem{RP10} A. Rao and N. Pinnawala,``A family of two-weight irreducible cyclic codes," 
\emph{IEEE Trans. Inform. Theory,} vol. 56, no. 6, pp. 2568--2570, 2010.   

 \bibitem{Shamir} A. Shamir, ``How to share a secret,'' {\em Comm. ACM,} 
         vol. 22, no. 11,  pp. 612--613, 1979.


\bibitem{YD06} J. Yuan and C. Ding, ``Secret sharing schemes from three classes of linear codes," 
{\em IEEE Trans. Inform. Theory,} vol. 52, no. 1, pp. 206--212, 2006.

\bibitem{Xia} Y. Xia, X. Zeng  and L. Hu, ``Further crosscorrelation properties of sequences
with the decimation factor $d = (p^n+1)/(p+1) +(p^n-1)/2$," 
{\em Appl. Algebra Eng. Commun. Comput.}, vol. 21,  pp. 329--342, 2010. 

\bibitem{YD04} J. Yuan and C. Ding, ``Secret sharing schemes from three classes  
         of linear codes,'' {\em IEEE Trans. Inform. Theory}, vol. 52, no. 1, pp. 206--212, 2006.

\bibitem{ZD13} 
Z. Zhou and C. Ding, ``A class of three-weight cyclic codes," \emph{Finite Fields Appl.}, vol.
25, pp. 79--93, 2014. 

\end{thebibliography}
\end{document}